\documentclass[11pt]{article}

\usepackage{fullpage}              
\usepackage{graphicx}              
\usepackage{amsmath}               
\usepackage{amsfonts}              
\usepackage{amsthm}                
\usepackage{amssymb}
\usepackage{mathrsfs}
\usepackage{url}
\usepackage{color}
\usepackage{hyperref}
\usepackage{algpseudocode}
\usepackage[plain]{algorithm}
\usepackage{enumitem}
\usepackage{booktabs}
\usepackage{authblk}
\usepackage{tcolorbox}

\hypersetup{
	unicode = true,
	colorlinks = true,
	citecolor = blue,
	filecolor = blue,
	linkcolor = blue,
	urlcolor = blue,
	pdfstartview = {FitH},
}

\theoremstyle{plain}
\newtheorem{theorem}{Theorem}
\newtheorem{lemma}[theorem]{Lemma}

\newtheorem{proposition}[theorem]{Proposition}
\theoremstyle{definition}

\newtheorem*{remark}{Remark}

\newtheorem*{problem}{Problem}
\newtheorem*{fact}{Fact}

\algrenewcommand{\Require}{\item[\textbf{Input:}]}
\algrenewcommand{\Ensure}{\item[\textbf{Output:}]}
\newcommand{\algbox}[1]{
	\begin{tcolorbox}[width = 0.8\textwidth, colback = white, arc = 2pt, boxrule = 0.5pt] 
		#1 
	\end{tcolorbox}
}

\newcommand{\lrang}[1]{\langle#1\rangle}
\newcommand{\ldbrac}[1]{\lvert#1\rangle}
\newcommand{\abs}[1]{\left\vert#1\right\vert}

\newcommand{\Romnum}[1]{\uppercase\expandafter{\romannumeral #1}}

\DeclareMathOperator{\lcm}{lcm} 

\DeclareMathOperator{\Aut}{Aut}

\DeclareMathOperator{\polylog}{polylog}

\def\C{\ensuremath{\mathbb{C}}}
\def\K{\ensuremath{\mathbb{K}}}

\def\R{\ensuremath{\mathbb{R}}}
\def\Z{\ensuremath{\mathbb{Z}}}
\def\F{\ensuremath{\mathbb{F}}}

\def\MM{\ensuremath{\mathsf{M}}}
\def\CC{\ensuremath{\mathsf{C}}}

\title{Toward an Optimal Quantum Algorithm for Polynomial Factorization over Finite Fields}

\author[1]{Javad Doliskani}
\affil[1]{\small Institute for Quantum Computing, University of Waterloo}

\date{}

\begin{document}

\maketitle

\begin{abstract}
	We present a randomized quantum algorithm for polynomial factorization over finite fields. For 
	polynomials of degree $n$ over a finite field $\F_q$, the average-case complexity of our 
	algorithm is an expected $O(n^{1 + o(1)} \log^{2 + o(1)}q)$ bit operations. Only for a 
	negligible subset of polynomials of degree $n$ our algorithm has a higher complexity of $O(n^{4 
	/ 3 + o(1)} \log^{2 + o(1)}q)$ bit operations. This breaks the classical $3/2$-exponent 
	barrier for polynomial factorization over finite fields \cite{guo2016alg}.
\end{abstract}


\section{Introduction}
\label{sec:intro}

Factoring polynomials over finite fields has been known to be randomized polynomial time in the 
seminal work of Berlekamp \cite{Berlekamp70}. Various improvements on polynomial factorization, 
over many decades, have been made since Berlekamp's work. Two major steps were taken by 
\cite{cantor1981new} and \cite{von1992computing}; in the latter, von zur Gathen and Shoup proposed 
an efficient way of computing traces and powers of the Frobenius map in the polynomial ring modulo 
the input polynomial. Their algorithm is quasi-quadratic in the degree of the polynomial to be 
factored. Kaltofen and Shoup proposed a baby-step giant-step technique combined with an efficient 
method of computing simultaneous modular compositions that led to the first subquadratic algorithm 
for polynomial factorization \cite{KaSh98}. The complexity of their algorithm is $O(n^{1.815}\log 
q)$ finite field operations. 

The best known algorithm, as of now, is an implementation of the Kaltofen-Shoup algorithm due to 
Kedlaya and Umans \cite{kedlaya2011fast}. They proposed a fast algorithm for modular composition 
that when plugged into the Kaltofen-Shoup algorithm leads to an algorithm with complexity 
$O(n^{3/2 + o(1)}\log^{1 + o(1)}q + n^{1 + o(1)}\log^{2 + o(1)}q)$ bit operations. A recent 
result of Doliskani \textit{et al.}\,\cite{doliskani2017drinfeld} achieves the same complexity by 
exploiting some geometric properties of Rank 2 Drinfeld modules. Despite much effort, improving the 
exponent $3 / 2$ has remained an open problem. Guo \textit{et al.}\,\cite{guo2016alg} proposed a 
set of algebraic problems that are equivalent to improving the exponent $3 / 2$ in polynomial 
factorization. In this paper, we propose a quantum algorithm that improves this bound. Our main 
result is as follows:
\begin{theorem}
	\label{thm:main}
	Let $\F_q$ be a finite field, where $q$ is a power of a prime $p$, and let $n$ be a positive 
	integer. Let $\mathcal{C}_n$ be the set of all polynomials of degree $n$ over $\F_q$. There is 
	a randomized quantum algorithm for polynomial factorization over $\F_q$, that for all but a 
	negligible subset $\mathcal{B}_n$ of polynomials in $\mathcal{C}_n$ runs in an expected 
	$O(n^{1 + o(1)} \log^{2 + o(1)}q)$ bit operations. For the subset $\mathcal{B}_n$ the algorithm 
	runs in an expected $O(n^{4 / 3 + o(1)} \log^{2 + o(1)}q)$ bit operations.
\end{theorem}
It follows from Theorem \ref{thm:main} that the average-case runtime of the proposed algorithm for 
polynomials of degree $n$ is $O(n^{1 + o(1)} \log^{2 + o(1)}q)$ bit operations. This is essentially 
optimal with respect to the degree of the input. For a complete polynomial factorization algorithm 
we follow the Cantor-Zassenhaus scheme which consists of three stages:
\begin{description}
	\item[SFF] Squarefree factorization: Given a polynomial $f \in \F_q[x]$, outputs a set of 
	squarefree polynomials $f_1, \dots, f_r$ such that $f = f_1 f_2^2 \cdots f_r^r$.
	\item[DDF] Distinct-degree factorization: Given a squarefree polynomial $f \in \F_q[x]$, 
	outputs $f^{[1]}, \dots, f^{[n]}$ such that $f^{[i]}$ is the product of all monic irreducible	
	factors of $f$ of degree $i$ for all $i = 1, \dots, n$. The inputs of this stage are the 
	outputs of the SFF stage.
	\item[EDF] Equal-degree factorization: Given a squarefree polynomial $f \in \F_q[x]$ such that
	all irreducible factors of $f$ have the same degree, outputs the irreducible factors of $f$. 
	The inputs of this stage are the outputs of the DDF stage.
\end{description}
The SFF stage cane be done using an algorithm of Yun \cite{yun1976square} which takes $O(n^{1 + 
o(1)}\log^{1 + o(1)}q + n\log^{2 + o(1)}q)$ bit operations. For the EDF stage the probabilistic 
algorithm of von zur Gathen and Shoup \cite{von1992computing} takes an expected $O(n^{1 + 
o(1)}\log^{2 + o(1)}q)$ bit operations. The bottleneck of polynomial factorization is the DDF 
stage for which the best known algorithm takes $O(n^{3/2 + o(1)}\log^{1 + o(1)}q + n^{1 + o(1)} 
\log^{2 + o(1)}q)$ bit operations \cite{kedlaya2011fast}. Therefore, to asymptotically improve the 
complete factorization algorithm we shall focus only on the DDF stage.

\paragraph{Complexity model.}
We will always count the number of bit operations in our complexity estimates. Two $n$-bit integers 
can be multiplied in $O(n\log n \log\log n) = O(n^{1 + o(1)})$ bit operations \cite{vzGG}. Given 
two polynomials $f, g$ of degree $n$ over a ring $R$, the product $fg$ can be computed in $\MM(n) = 
O(n\log n \log\log n) = O(n^{1 + o(1)})$ operations in $R$ \cite{vzGG}. Sometimes it is more 
convenient to count operations in the base field $\F_q$. To convert to bit operations we always 
assume that $\F_q$ is represented by a quotient $\F_p[y] / h(y)$ for some polynomial $h \in \F_p[y]$ 
of degree $m$. The product of two elements of $\F_q$ can then be computed in $O(m^{1 + o(1)}\log^{1 
+ o(1)}p) = O(\log^{1 + o(1)}q)$ bit operations. For polynomials $f, g, h \in \F_q[x]$ of degree 
$n$, the modular polynomial composition $f(g) \bmod h$ can be done in $\CC(n) = O(n^{1 + o(1)} 
\log^{1 + o(1)}q)$ bit operations \cite{kedlaya2011fast}. Other operations such as $\gcd$ and 
reduction of polynomials of degree $n$ over $\F_q$ can be done in $O(n^{1 + o(1)} \log^{1 + 
o(1)}q)$ bit operations \cite{vzGG}. 

Given a polynomial $f \in \F_q[x]$ of degree $n$, define $\K = \F_q[x] / f$. The Frobenius 
endomorphism $\pi: \K \rightarrow \K$ is an $\F_q$-homomorphism defined by $x \mapsto x^q$. Given 
$\pi(x)$ and any integer $j > 0$, the power $\pi^{j}(x)$, which is $j$ successive compositions, can 
be computed using a simple binary-powering algorithm at the cost of $O(\CC(n) \log j) = O(n^{1 + 
o(1)} (\log j) \log^{1 + o(1)}q)$ bit operations. The polynomial $\pi(x) = x^q \bmod f$ can itself 
be computed in $O(n^{1 + o(1)} \log^{2 + o(1)}q)$ bit operations.

\paragraph{Our technique.}
The exponent $3/2$, achieved by the best previous algorithms for polynomial factorization, seems to 
be a natural outcome of the so called baby-step giant-step methodology. We take a completely 
different approach here. Let $f \in \F_q[x]$ be a squarefree monic polynomial of degree $n$, and let 
$d$ be the degree of the splitting of $f$ over $\F_q$. Then the degrees of the irreducible factors 
of $f$ divide $d$. We show that given $d$, one can efficiently compute a distinct-degree 
factorization of $f$. This reduces polynomial factorization to computation of splitting degrees. 
Computing the splitting degree of $f$ is equivalent to computing the order of the Frobenius 
automorphism $\pi$, defined above, in the automorphism group $\Aut(\K / \F_q)$. To compute the order 
of $\pi$, we use a quantum period finding algorithm.


\section{Estimating the order of an automorphism}
\label{sec:ord-frob}

Given a squarefree polynomial $f \in \F_q[x]$ of degree $n > 1$, let $\K = \F_q[x] / f$. The 
$\F_q$-endomorphism $\pi: \K \rightarrow \K$ defined by $x \mapsto x^q$ is called the Frobenius 
endomorphism. Since $f$ is squarefree, $\pi$ is an automorphism. The cyclic group of automorphisms 
$\lrang{\pi}$ is finite. We give a quantum algorithm that efficiently estimates the order of any 
automorphism $\sigma \in \lrang{\pi}$. The quantum algorithm is not new; it is a standard order 
finding algorithm adapted to our situation. 

To find the order of an automorphism $\sigma$ we use the techniques in \cite{kaye2007introduction, 
nielsen2010quantum}. The group $\lrang{\sigma}$ is isomorphic to the additive group $\Z/r\Z$ for 
some ineteger $r > 0$. Since the action of $\sigma$ on $\K$ is determined by the action of $\sigma$ 
on $x$, the powers $\sigma^j$ are represented by the polynomials $\sigma^j(x) \in \K$ for all $j 
\ge 0$. A polynomial $h \in \K$ is represented using an array of size $n$ containing the 
coefficients of $h$. The number of qubits for representing the elements of $\K$ is then $n\lceil 
\log q \rceil$. For an integer $j$ and a polynomial $h \in \K$, define the action of $\sigma$ on the 
state $\ldbrac{j} \ldbrac{h}$ as
\begin{equation}
	\label{equ:sig-act}
	\sigma \ldbrac{j}\ldbrac{h} = \ldbrac{j}\ldbrac{\sigma^j(x) \oplus h}
\end{equation}
where $\oplus$ is simply the xor of two qubit arrays.

The main ingredient of order finding algorithms is the \textit{quantum Fourier transform} (QFT). 
For any finite group $G$ and any function $t: G \rightarrow \C$ the QFT over $G$ is a specific 
unitary operator on the vector space $\C[G]$ that takes the complex numbers $\{ t(g) \}_{g \in 
G}$ to another set of $\abs{G}$ complex numbers \cite{hallgren2003hidden}. We shall only need QFT 
over the additive group $\Z/N\Z$, denoted by $F_N$, where $N = 2^m$ for some integer $m$. In this 
case, we have 
\[ F_N: \ldbrac{k} \longmapsto \frac{1}{\sqrt{N}}\sum_{j \in \Z/N\Z}\zeta_N^{kj}\ldbrac{j} \]
where $k \in \Z/N\Z$ and $\zeta_N = e^{2\pi i / N} \in \C$ is a primitive $N$-th root of unity. 
Given an integer $\ell > 0$, an $\ell$-bit estimate of the order $r$ of $\sigma$ is computed as 
follows. Prepare two registers with initial value $\ldbrac{0}\ldbrac{0}$, the first register of 
length $m = 2\ell + 1$ qubits and the second of length $n\lceil \log q \rceil$ qubits. Create a 
superposition in the first register to get the state 
\[ \ldbrac{\psi} = \frac{1}{\sqrt{N}} \sum_{j \in \Z/N\Z} \ldbrac{j}\ldbrac{0}. \]
Applying $\sigma$ to $\ldbrac{\psi}$ and rewriting the resulting sum based on the period $r$ gives
\[ \sigma\ldbrac{\psi} = \frac{1}{\sqrt{N}} \sum_{j \in \Z/N\Z} \ldbrac{j}\ldbrac{\sigma^j(x)} = 
\sum_{b = 0}^{r - 1}\left( \frac{1}{\sqrt{N}}\sum_{z = 0}^{m_b - 1}\ldbrac{zr + b} \right) 
\ldbrac{\sigma^j(x)} \]
where $m_b$ is the largest integer such that $m_b \le (N - b - 1) / r + 1$ 
\cite{kaye2007introduction}. Discarding the second register leaves the first register in the 
superposition state
\[ \frac{1}{\sqrt{m_b}} \sum_{z = 0}^{m_b - 1}\ldbrac{zr + b} \]
where $b \in \Z/r\Z$ is selected nearly uniformly at random. Applying $F_N^{-1}$ and measuring, we 
obtain an integer $k$ such that $k / N$ is an estimate of $j / r$ for some nearly uniformly random 
$j \in \Z/r\Z$. More precisely, one can show that $\abs{k / N - j / r} \le \frac{1}{2N}$ for some 
nearly uniformly random $j \in \Z/r\Z$. Using rational number reconstruction, we can obtain 
integers $j_1, r_1$ such that $j_1 / r_1 = j / r$. Repeating this process, we obtain another pair 
$j_2, r_2$. It can be shown that $r = \lcm(r_1, r_2)$ with probability $\approx 6 / \pi^2$. 
Therefore, we only need to run the above procedure a constant number of times to obtain $r$ with 
high probability. This is summarized in Algorithm \ref{alg:qtm-order}.

\begin{algorithm}[t]
	\caption{Estimate the order of an automorphism}
	\label{alg:qtm-order}
	\centering
	\algbox{
	\begin{algorithmic}[1]
		\Require 
		\item[-] A squarefree monic polynomial $f \in \F_q[x]$ of degree $n$
		\item[-] An automorphism $\sigma \in \lrang{\pi}$ of the $\F_q$-algebra $\K = \F_q[x] / f$
		\item[-] An integer $\ell < n$ as a bound for the number of bits by which the order of the 
		automorphism $\sigma$ is estimated
		\Ensure The order $r$ of $\sigma$ or `Fail'
		\State $m := 2\ell + 1$, $N := 2^m$
		\State\label{step:superpos}Prepare registers $\ldbrac{0}\ldbrac{0}$ of length $m$ qubits 
		and $n\lceil \log q \rceil$ qubits respectively.
		\State\label{step:app-sig}Create a superposition in the first register and apply $\sigma$ 
		to get
		\[ \frac{1}{\sqrt{N}} \sum_{j \in \Z/N\Z} \ldbrac{j}\ldbrac{\sigma^j(x)} \]
		\State\label{step:qft-inv}Discard the second register and apply $F_N^{-1}$ to the first 
		register to obtain an integer $k$ such that $j_1 / r_1 = k / N$ is an estimate of $j / r$ 
		for some nearly uniformly random $j \in \Z/r\Z$
		\State\label{step:frac-est}Compute $r_1$ using rational number reconstruction
		\State Repeat Steps \ref{step:superpos}-\ref{step:frac-est} to obtain another $r_2$, 
		and compute $r := \lcm(r_1, r_2)$
		\State\label{step:sig-test} If $\sigma^r = \text{id}$ then return $r$, else return `Fail'
	\end{algorithmic}}
\end{algorithm}
	
\begin{proposition}
	Given a squarefree monic polynomial $f \in \F_q[x]$ of degree $n$, an element $\sigma \in 
	\lrang{\pi}$ and an integer $\ell > 0$, Algorithm \ref{alg:qtm-order} computes the order of 
	$\sigma$, or fails, with probability $O(1)$ and in $O(\ell n^{1 + o(1)}\log^{1 + o(1)}q)$ bit 
	operations.
\end{proposition}
\begin{proof}
	The quantum Fourier transform $F_N$ and its inverse $F_N^{-1}$ in Steps \ref{step:app-sig} 
	and \ref{step:qft-inv} are done in $O(\log N \log\log N) = O(\ell\log\ell)$ bit operations 
	\cite{cleve1998quantum, childs2010quantum}. For an ineteger $j$, the power $\sigma^j$ in Step 
	\ref{step:app-sig} is computed in $O(\CC(n)\log j)$ operations in $\F_q$. Since $j < 2^\ell$ 
	and $\CC(n) \in O(n^{1 + o(1)}\log^{1 + o(1)}q)$, Step \ref{step:app-sig} is performed in 
	$O(\ell n^{1 + o(1)}\log^{1 + o(1)}q)$ bit operations.
	
	The rational number reconstruction of Step \ref{step:frac-est} is done at the cost of 
	$O(\log^{1 + o(1)}N) = O(n^{1 + o(1)})$ bit operations \cite{pan2002acceleration}. Since $r < 
	2^\ell$, Step \ref{step:sig-test} is done using $O(\ell n^{1 + o(1)}\log^{1 + o(1)}q)$ bit 
	operations. Adding these together establishes the claimed complexity.
\end{proof}


\section{Computing the order of the Frobenius}

Let $f \in \F_q[x]$ be a squarefree polynomial of degree $n$ and let $\K = \F_q[x] / f$. Let $f = 
f_1f_2 \cdots f_k$ be the factorization of $f$ into distinct irreducible factors. By the Chinese 
Remainder Theorem, there is an isomorphism of rings
\begin{equation}
\label{equ:crt}
	\K \cong \bigoplus_{i = 1}^k \F_q[x] / f_i(x).
\end{equation}
Let $d_i = \deg(f_i)$ for $i = 1, \dots, k$. Then the splitting field of $f$ has degree $d = 
\lcm(d_1, \dots, d_k)$ over $\F_q$. Also note that the Frobenius automorphism $\pi: x \mapsto x^q$ 
of $\K$ is the coproduct of the Frobenius automorphisms $\pi_i: x \mapsto x^q$ of the fields 
$\F_q[x] / f_i$. It follows that the order of the group $\lrang{\pi}$ is $d$ as well.

Given $\sigma \in \lrang{\pi}$ of order $d'$, one could compute the exact value of $d'$ using 
Algorithm \ref{alg:qtm-order} as follows. Start with a small estimation bound $\ell$ and call the 
algorithm with input $(f, \sigma, \ell)$. If the algorithm does not output `Fail' then stop, 
otherwise set $\ell := \ell + 1$ and repeat. This results in too many calls to Algorithm 
\ref{alg:qtm-order} if $d'$ is too large. It is therefore crucial to know a reasonable bound on the 
size $d$ of the group $\lrang{\pi}$. To obtain such a bound, one naturally looks at the distinct 
values of the $d_i$'s above. It turns out that for almost all polynomials $f \in \F_q[x]$ of degree 
$n$ the degree of the splitting field of $f$ over $\F_q$ is $O(2^{\polylog(n)})$, that is $\log d 
\in O(\polylog(n))$. In fact, Knopfmacher proves the following rather precise estimate for the 
number of distinct values of the $d_i$'s.
\begin{theorem}[\cite{knopfmacher1999degrees}]
	\label{thm:d-mean}
	The number of irreducible factors of a polynomial $f \in \F_q[x]$ of degree $n$ has mean value 
	$\log n + O(1)$ and variance $\log n + O(1)$ as $n \rightarrow \infty$. In particular, almost 
	all polynomials of degree $n$ have approximately $\log n$ irreducible factors.
\end{theorem}
From Theorem \ref{thm:d-mean} we see that for almost all polynomials $f$ of degree $n$, if $d_1, 
\dots, d_k$ are the irreducible degrees of $f$ then 
\begin{equation}
	\label{equ:ag-mean}
	d = \lcm(d_1, \dots, d_k) \le \left( \frac{1}{k} \sum_{i = 1}^rd_i \right)^k \le \left( 
	\frac{n}{k} \right)^k \approx \left( \frac{n}{\log n} \right)^{\log n} = O(2^{\log^2n})
\end{equation}
where the first inequality is the arithmetic-geometric mean. Dixon and Panario 
\cite{dixon2004degree} give a stronger statement for bounds on $d$ based on the analogy between 
degrees of splitting fields of polynomials of degree $n$ and orders of elements in the 
\textit{symmetric group} $S_n$ \cite{erdos1965some, erdos1967some}. We state their result here 
for completeness. Following the notation of \cite{dixon2004degree}, let $\lambda$ be a partition of 
$n$ of shape $[1^{k_1} 2^{k_2} \dots n^{k_n}]$, i.e., there are $k_i$ parts of size $i$. Define 
$m(\lambda)$ to be the $\lcm$ of the $i$'s where $k_i \ne 0$. For any $t > 0$ define 
\[ \Phi_n(t) = \{ \lambda \vdash n : \abs{\log m(\lambda) - \log^2n / 2} > t\log^{3/2}n / \sqrt{3} 
\} \]
where $\lambda \vdash n$ denotes a partition $\lambda$ of $n$. A monic polynomial $f \in \F_q[X]$ 
of degree $n$ is said to be of shape $\lambda$ if the degrees of the irreducible factors of $f$ are 
the parts of $\lambda$. We have
\begin{theorem}[{\cite[Theorem 1]{dixon2004degree}}]
\label{thm:split-order}
	For each $\lambda \vdash n$, let $w(\lambda, q)$ denote the portion of monic polynomials of 
	degree $n$ over $\F_q$ that are of shape $\lambda$. Then there exists a constant $c > 0$ 
	such that for each $t \ge 1$ there exists an integer $n_0(t)$ such that 
	\[ \sum_{\lambda \in \Phi_n(t)} w(\lambda, q) \le ce^{-t / 4} \text{ for all } q \text{ and all 
	} n \ge n_0(t). \]
	In particular, almost all monic polynomials of degree $n$ have splitting fields of degree at 
	most $\exp((\frac{1}{2} + o(1))\log^2n)$ over $\F_q$ as $n \to \infty$.
\end{theorem}
It follows from the above that for almost all squarefree $f$, the order $d$ of the Frobenius 
$\sigma$ of $\K$ can be computed in $O(n^{1 + o(1)}\log^{1 + o(1)}q)$ bit operations: we simply 
call Algorithm \ref{alg:qtm-order} with the bound $\ell = \log^2n$. There are, however, some $f$ 
for which $d$ can be as large as $\exp(c\sqrt{n \ln n})$ for some constant $c > 0$  
\cite{erdos1965some}. In this case, the input bound has to be $\ell = c\sqrt{n \ln n}$, so 
Algorithm \ref{alg:qtm-order} computes $d$ in $O(n^{3/2 + o(1)}\log^{1 + o(1)}q)$ bit operations. 
This already dashes the hopes of improving upon the classical bound if we ever wanted to use 
Algorithm \ref{alg:qtm-order} as a subroutine in a polynomial factorization algorithm! Fortunately, 
this problem can be circumvented by preprocessing $f$ before calling Algorithm \ref{alg:qtm-order}. 
We need the following lemma.
\begin{lemma}
\label{lem:frob-sub-ord}
	Let $f \in \F_q[x]$ be a squarefree polynomial of degree $n$. Let $\tilde{f}$ be the product of 
	the factors of $f$ with degree higher than $n^{2 / 3}$. Then the order of the Frobenius 
	automorphism $\sigma$ of $\F_q[x] / \tilde{f}$ is less than $\exp(\frac{2}{3}\sqrt[3]{n} \ln 
	n)$.
\end{lemma}
\begin{proof}
	Let $\tilde{f} = \tilde{f}_1\tilde{f}_2 \cdots \tilde{f}_r$ be the irreducible factorization of 
	$\tilde{f}$. Then $d_i = \deg(\tilde{f}_i) > n^{2 / 3}$ and hence $r < n^{1 / 3}$. The degree 
	of the Frobenius automorphism $\sigma$ of $\F_q[x] / \tilde{f}$ is $d = \lcm(d_1, d_2, \dots, 
	d_r)$, which is also the degree of the splitting field of $\tilde{f}$ over $\F_q$. We have, as 
	in 
	\eqref{equ:ag-mean},
	\[ d \le \prod_{i = 1}^rd_i \le \left( \frac{1}{r} \sum_{i = 1}^rd_i \right)^r \le \left( 
	\frac{n}{r} \right)^r. \]
	The function $y(t) = (n / t)^t$ on $\R^{> 0}$ has a global maximum at $t = n / e$. Therefor, 
	for $r < n^{1/3}$ we have
	\[ \left( \frac{n}{r} \right)^r < \left( n^{2 / 3} \right)^{n^{1 / 3}} = \exp\left( \frac{2}{3} 
	\sqrt[3]{n} \ln n \right). \qedhere \]
\end{proof}
Given an $f$ for which $d \in O(\exp(c\sqrt{n \ln n}))$, we can first extract all irreducible 
factors of $f$ of degree at most $n^{2 / 3}$. This can be done using the algorithm of \cite[\S 
8]{kedlaya2011fast} at the cost of $O(n^{4 / 3 + o(1)}\log^{2 + o(1)}q)$ bit operations. For the 
remaining polynomial $\tilde{f}$, the order $\tilde{d}$ of the Frobenius $\tilde{\sigma}$ of 
$\F_q[x] / \tilde{f}$ is at most $\exp(\frac{2}{3}\sqrt[3]{n} \ln n)$ by Lemma 
\ref{lem:frob-sub-ord}. Now, if we call Algorithm \ref{alg:qtm-order} with input $(\tilde{f}, 
\tilde{\sigma}, \sqrt[3]{n} \ln n)$, we obtain $\tilde{d}$ at the cost of $O(n^{4 / 3 + o(1)}\log^{1 
+ o(1)}q)$ bit operations.

Since $d$ is not known a priori, to compute the exact value of $d$ using Algorithm 
\ref{alg:qtm-order}, we could set $\ell = c\sqrt{n \log n}$, the maximum possible bound. But, as 
mentioned above, this results in the exponent $3 / 2$ which is not better than the classical 
exponent for polynomial factorization. Instead, we do the following. We first call Algorithm 
\ref{alg:qtm-order} with the bound $\ell = \log^2n$ (or $\ell = \log^cn$ for any reasonable constant 
$c \ge 2$). If the output is not `Fail' then we are done. Otherwise, we remove all factors of degree 
$\le n^{2 / 3}$ from $f$ and call Algorithm \ref{alg:qtm-order} with the bound $\ell = \sqrt[3]{n} 
\log n$. These remarks are summarized in Algorithm \ref{alg:order}. Note that Algorithm 
\ref{alg:order} also accepts a bound $\ell$ as an extra parameter. This increases the flexibility 
of the algorithm that proves useful in later stages of our factorization algorithm, see Algorithm 
\ref{alg:ddf}.

\begin{algorithm}[t]
	\caption{Compute the order of a power of the Frobenius}
	\label{alg:order}
	\centering
	\algbox{
	\begin{algorithmic}[1]
		\Require 
		\item[-] A monic squarefree polynomial $f \in \F_q[x]$ of degree $n$
		\item[-] An element $\sigma \in \lrang{\pi}$ of the Frobenius group of $\F_q[x] / f$
		\item[-] A bound $\ell$
		\Ensure A monic polynomial $g \mid f$ and the order of the Frobenius of $\F_q[x] / g$
		\State Compute an estimate $d$ of $\sigma$ using Algorithm \ref{alg:qtm-order} with input 
		$(f, \sigma, \ell)$
		\If {the output is not `Fail'}
			\State \Return $f, d$
		\EndIf
		\State Remove and output all irreducible factors of degree $\le n^{2 / 3}$ from 
		$f$, and let $\tilde{f}$ be the resulting polynomial
		\State Set $\tilde{\ell} := \lceil c_1\sqrt[3]{n} \log n \rceil$ for a suitable constant 
		$c_1$
		\State Compute an estimate $\tilde{d}$ of the Frobenius $\tilde{\sigma}$ of $\F_q[x] / 
		\tilde{f}$ using Algorithm \ref{alg:qtm-order}.
		\State \Return $\tilde{f}, \tilde{d}$
	\end{algorithmic}}
\end{algorithm}

\begin{proposition}
	\label{prop:exact-d}
	Let $f \in \F_q[x]$ be a squarefree polynomial of degree $n$ and let $d$ be the degree of the 
	splitting field of $f$ over $\F_q$. Given a bound $\ell$, Algorithm \ref{alg:order} runs in
	\begin{itemize}
		\item $O(\ell n^{1 + o(1)}\log^{1 + o(1)}q)$ bit operations if $d \in O(2^\ell)$, or
		\item $O(n^{4 / 3 + o(1)}\log^{2 + o(1)}q)$ bit operations otherwise.
	\end{itemize}
\end{proposition}
\begin{proof}
	Follows from the previous remarks.
\end{proof}


\section{An algorithm for distinct-degree factorization}

In this section, we give a dynamic programing algorithm for the distinct-degree factorization of a 
squarefree polynomial $f(x)$ of degree $n$ over $\F_q$. Our algorithm invokes the quantum algorithm 
of the previous section in order to determine the order of a power of the Frobenius modulo $f$. We 
aim to solve the following subproblem.
\begin{problem}
	Given a tuple $(f(x), s)$ where $f(x)$ is a squarefree polynomial and $s > 0$ is an integer 
	that divides the degrees of all irreducible factors of $f$, produce a set $T = \{ (f_1, s_1), 
	\dots, (f_k, s_k) \}$ of tuples such that 
	\begin{enumerate}
		\item\label{cond:split} The polynomials $\{ f_i \}_{1 \le i \le k}$ are nontrivial 
		splitting of $f$ unless $f$ has only irreducible factors of degree $s$ (i.e., it is already 
		a distinct-degree part),
		\item\label{cond:div} $s \mid s_i$ for all $1 \le i \le k$.
	\end{enumerate}
\end{problem}
To obtain a DDF of $f$, we start with the tuple $(f, 1)$ and recursively solve the above problem 
for every output tuple until each tuple is a distinct-degree part of $f$. Before discussing the 
solution to the above problem, we recall the following key fact.
\begin{fact}
	For any integer $d > 0$, the polynomial $x^{q^d} - x \in \F_q[x]$ is the product of all monic 
	irreducible polynomials whose degree divide $d$. 
\end{fact}
This means that $\gcd(x^{q^d} - x, f)$ is the product of the distinct-degree parts of $f$ whose 
degree divide $d$.

Now, let $(f(x), s)$ be given as above and let $d_1, \dots, d_k$ be the degrees of irreducible 
factors of $f$. Let $\sigma$ be the Frobenius of $\F_q[x] / f$, and set $T = \{\}$. Compute the 
Frobenius power $\tilde{\sigma} = \sigma^s$. If $\tilde{\sigma} = \text{id}$ then $d_i \mid s$ for 
all $1 \le i \le k$. But since $s \mid d_i$ for all $1 \le i \le k$ by assumption, we have $d_1 = 
\cdots = d_k = s$ so $f$ is a distinct-degree part and we are done. So assume $\tilde{\sigma} \ne 
\text{id}$. Compute the order $d$ of $\tilde{\sigma}$ using Algorithm \ref{alg:order}. Note that 
the output of Algorithm \ref{alg:order} might be some distinct-degree parts of $f$ and a polynomial 
$\tilde{f} \mid f$ such that $d$ is the order of $\tilde{\sigma}$ as an automorphism of $\F_q[x] / 
\tilde{f}$. We assume, without loss of generality, that the distinct-degree outputs are appended to 
$T$, $f := \tilde{f}$ and $d_1, \dots, d_k$ are again the degrees of irreducible factors of $f$.
Write $d = p_1^{e_1} p_2^{e_2} \cdots p_\ell^{e_\ell}$, where $p_i$'s are pairwise 
distinct primes and $e_i$'s are positive integers. Since $sd = \lcm(d_1, \dots, d_k)$, each prime 
power $p_i^{e_i}$ must divide at least one of the $d_j$'s. Let 
\[ g(x) := \gcd(\tilde{\sigma}^{d / \prod_{i = 1}^\ell p_i}(x) - x, f). \]
Then $g(x)$ is the product of all factors of $f$ with degrees dividing $s \prod_{i = 1}^\ell 
p_i^{e_i - 1}$, and $g_0(x) := f / g$ is the product of factors with degrees a multiple of at least 
one of the $sp_i^{e_i}$. If $g \ne 1$ then add the tuple $(g, s)$ to $T$. To separate all factors 
with degrees a multiple of $sp_1^{e_1}$, we first compute
\[ g_1(x) = \gcd(\tilde{\sigma}^{d / p_1}(x) - x, g_0) \]
and then $h_1 = g_0 / g_1$. If $g_1$ is a nontrivial factor of $g_0$ then add the tuple $(h_1, 
sp_1^{e_1})$ to $T$. Otherwise we have $g_1 = 1$, so $sp_1^{e_1}$ divides the degrees of all
irreducible factors of $g_0$. In this case we replace $g_1$ and $s$ with $g_0$ and $sp_1^{e_1}$ 
respectively. Next, we separate all factors of $g_1$ with degree a multiple of $sp_2^{e_2}$ by 
computing
\[ g_2(x) = \gcd(\tilde{\sigma}^{d / p_2}(x) - x, g_1) \]
and repeating the same process. Doing this for all $i = 2, \dots, \ell$, we obtain a list of tuples
\begin{equation}
	\label{equ:tuple}
	T = \{ (h_{i_1}, s_{i_1}), \dots, (h_{i_r}, s_{i_r}) \}
\end{equation}
such that $h_{i_j}$ is a nontrivial factor of $f$, $s \mid s_{i_j}$ and $s_{i_j}$ divide all 
irreducible degrees of $h_{i_j}$. We have, by construction, $f = h_{i_1}h_{i_2} \cdots h_{i_r}$, so 
the list of tuple $T$ in \eqref{equ:tuple} satisfy conditions \ref{cond:split} and \ref{cond:div} in 
the above problem. Algorithm \ref{alg:ddf} uses this procedure to recursively compute a DDF of a 
given polynomial $f$. We need to prove that the recursion is not too deep. This is established by 
Lemma \ref{lem:depth}.

\begin{algorithm}[t]
	\caption{DDF}
	\label{alg:ddf}
	\centering
	\algbox{
	\begin{algorithmic}[1]
		\Require 
		\item[-] A monic squarefree polynomial $f \in \F_q[x]$ of degree $m$
		\item[-] An integer $s > 0$ that divides the degrees of all irreducible factors of $f$
		\item[-] An integer $n > 0$
		\Ensure Distinct-degree factors of $f$
		\State\label{step:sigma-s}Compute $\tilde{\sigma} := \sigma^s \bmod f$
		\If {$\tilde{\sigma} = $ id}
			\State Output $f$ and return
		\EndIf
		\State\label{step:order}Compute the order $d$ of $\tilde{\sigma} \bmod f$ using Algorithm 
		\ref{alg:order} with inputs $(f, \tilde{\sigma}, \log^2n)$, and let $\tilde{f}$ be the 
		output polynomial
		\State\label{step:factor-d}Factor $d$ and let $d := p_1^{e_1} p_2^{e_2} \cdots 
		p_\ell^{e_\ell}$
		\State\label{step:gcd-first}Compute $g(x) := \gcd(\tilde{\sigma}^{d / \prod_{i = 1}^\ell 
		p_i}(x) - x, f)$ and $g_0(x) := f / g$
		\State $T := \{\}$
		\State\label{step:proper}If $g \ne 1$ then add $(g, s, n)$ to $T$
		\For {$i = 1$ to $\ell$ }\label{step:for-ddf}
			\State Compute $g_i(x) := \gcd(\tilde{\sigma}^{d / p_i}(x) - x, g_{i - 1})$
			\If {$g_i \ne 1$} 
				\State Compute $h_i := g_{i - 1} / g_i$ and add the tuple $(h_i, sp_i^{e_i}, n)$ to 
				$T$
			\Else
				\State $g_i := g_{i - 1}$, $s := sp_i^{e_i}$
			\EndIf
		\EndFor
		\State Add $(g_\ell, s, n)$ to $T$
		\State Recursively process all tuples in $T$
	\end{algorithmic}}
\end{algorithm}

\begin{lemma}
	\label{lem:depth}
	Given a squarefree polynomial $f \in \F_q[x]$ of degree $n$, the recursion depth of Algorithm 
	\ref{alg:ddf} for the input $(f, 1, n)$ is $O(\log n)$.
\end{lemma}
\begin{proof}
	Let $d'$ be an irreducible degree of $f(x)$ and let $(f(x), 1) = (a_1(x), s_1), (a_2(x), s_2), 
	\dots, (a_k(x), s_k) = (a_k(x), d')$ be the path to the degree $d'$ part of $f$ generated by 
	Algorithm \ref{alg:ddf}. Let $d_i$ be the degree of the Frobenius $\sigma^{s_i} \bmod a_i(x)$ 
	computed in Step \ref{step:order} for $i = 1, \dots, k$, and let $d_1 := p_1^{e_1} p_2^{e_2} 
	\cdots p_\ell^{e_\ell}$ be the prime factorization of $d_1$ at Step \ref{step:factor-d}. Write 
	$d' = p_1^{\alpha_1} \cdots p_\ell^{\alpha_\ell}$ where $0 \le \alpha_i \le e_i$ for all $i 
	= 1, \dots, \ell$. Now, given $1 \le j \le k$, write $d_j = p_1^{\beta_1} \cdots 
	p_\ell^{\beta_\ell}$. Then only two things can happen from $(a_j(x), s_j)$ to $(a_{j + 1}(x), 
	s_{j + 1})$:
	\begin{enumerate}
		\item\label{case:proper} $\alpha_r < \beta_r$ for all $r = 1, \dots, \ell$. In other words, 
		non of the prime powers $p_i^{\beta_i}$ divides $d'$. This happens when $(a_{j + 1}(x), s_{j 
		+ 1})$ is the tuple in Step \ref{step:proper}. In this case, $s_{j + 1} = s_j$ and $d_j \mid 
		d_{j + 1} / \prod_{i = 1}^\ell p_i = p_1^{\beta_1 - 1} \cdots p_\ell^{\beta_\ell - 1}$.
		\item\label{case:nonprop} $\alpha_r = \beta_r$ for at least one $1 \le r \le \ell$. In 
		other words, $d'$ is a multiple of at least one of the prime powers $p_i^{\beta_i}$. This 
		happens when $(a_{j + 1}(x), s_{j + 1})$ is one of the tuples in Step \ref{step:for-ddf}. 
		In this case, $ts_j \mid s_{j + 1}$ for some integer $t > 1$. 
	\end{enumerate}
	Let $e = \max_{1 \le i \le \ell}\{ e_i \}$. Since $d_i \le d_1$ for all $i = 1, \dots, k$, case 
	\ref{case:proper} can happen at most $e \le \log n$ number of times. Since $s_i \le d' \le n$ 
	for all $i = 1, \dots, k$, case \ref{case:nonprop} can also happen at most $\log n$ number of 
	times. Therefore, we always have $k \in O(\log n)$. 
\end{proof}
Given a squarefree polynomial $f \in \F_q[x]$ of degree $n$, calling Algorithm \ref{alg:ddf} with 
the input tuple $(f, s = 1, n)$ will produce the DDF of $f$. The auxiliary input integer $n$ is 
never changed. It is used in Step \ref{step:order} to input the bound $\ell = \log^2n$ to Algorithm 
\ref{alg:order}. This will simplifies the complexity analysis of the algorithm. Note that when the 
algorithm is called for the first time, if $d \in O(2^{\log^2n})$ then we could always assume that 
this bound holds for the subsequent $d$'s in the the next stages of the recursion: for any input 
polynomial $g$ of degree $m$ in an intermediate stage we have $m < n$ and $g \mid f$ so that the 
order of the Frobenius of $\F_q[x] / g$ is always bounded by that of $\F_q[x] / f$. Let us now 
analyze the complexity of Algorithm \ref{alg:ddf}. We shall need the following two propositions.
\begin{proposition}
	\label{prop:frobs}
	Given a polynomial $f \in \F_q[x]$ of degree $n$, a power $\tilde{\sigma}(x) = x^{q^t}$ of the 
	Frobenius of $\F_q[x] / f$, and an integer $d = p_1^{e_1} p_2^{e_2} \cdots p_\ell^{e_\ell}$ in 
	the factored form, the sequence
	\begin{equation}
	\label{equ:frobs}
		\tilde{\sigma}^{d / p_1}(x) \bmod f, \dots, \tilde{\sigma}^{d / p_\ell}(x) \bmod f
	\end{equation}
	can be computed in $O(\CC(n)(\log \ell)(\log d) + \ell\MM(n))$ operations in $\F_q$.
\end{proposition}
\begin{proof}
	Denote by $\{ d, [1, \ell], \tilde{\sigma} \}$ the problem of computing $\tilde{\sigma}^{d / 
	p_i}(x) \bmod f$ for all $i$ in the range $[1, \ell]$. We recursively solve this problem as 
	follows. Let $r = \lfloor \ell / 2 \rfloor$, $d_1 = p_{r + 1}^{e_{r + 1}} \cdots 
	p_\ell^{e_\ell}$ and $d_2 = p_1^{e_1} \cdots p_r^{e_r}$. Then the problem is reduced to the 
	two subproblems $\{ d_2, [1, r], \tilde{\sigma}^{d_1} \}$ and $\{ d_1, [r + 1, \ell], 
	\tilde{\sigma}^{d_2} \}$. The problem $\{ p_i^{e_i}, [i, i], \tilde{\sigma}^{d / p_i^{e_i}} \}$ 
	is solved by simply raising $\tilde{\sigma}^{d / p_i^{e_i}}$ to the power $p_i^{e_i - 1}$.
	
	Computing the powers of the Frobenius is done using modular composition. At any level of the 
	recursions, the total number of modular compositions is $O(\log d)$. Therefore, at level $j$ we 
	spend $O(\CC(n) \log d)$ operations in $\F_q$ on modular compositions, and $O(2^j\MM(n))$ 
	operations in $\F_q$ on $2^j$ polynomial reductions. Since the depth of the recursion is $\log 
	\ell$, the 	claimed runtime follows.
\end{proof}
\begin{proposition}
	\label{prop:factor-d}
	The integer factorization in Step \ref{step:factor-d} of Algorithm \ref{alg:ddf} can be done in 
	$O(n^{1 + o(1)})$ bit operations.
\end{proposition}
\begin{proof}
	We know that $d$ is $n$-smooth, i.e., $p_i \le n$ for all $i = 1, \dots, \ell$. If $d \in 
	O(2^{\log^2n})$ then $d$ can be completely factored in $O(n^{1 + o(1)})$ bit operations using a 
	naive trial division algorithm. So we are left with the case $d \in O(2^{c\sqrt[3]{n} \log 
	n})$. For this we use a \textit{subproduct tree}.
	
	Let $P = \{ q_1, q_2, \dots, q_s \}$ be the set of all primes $\le n$. The subproduct tree for 
	the set $P$ is a binary tree described as follows. The root of the tree is the product $q_1 q_2 
	\dots q_s$. Now divide the set $P$ into two halves for the left and right subtrees with roots 
	the products $q_1 q_2 \dots q_{\lceil s / 2 \rceil}$ and $q_{\lceil s / 2 \rceil + 1} \dots 
	q_s$, respectively. Repeating this this process results in a tree with leaves the primes in 
	$P$. The tree can be built recursively: starting from the leaves, the value of a parent node 
	is the product of the roots of the left and right subtrees. Since two $m$-bit integers are 
	multiplied in $O(m^{1 + o(1)})$ bit operation, it is easily seen that each level of the tree is 
	built at the cost of $O(n^{1 + o(1)})$ bit operations. Since the height of the tree is 
	$O(\log n)$, the total cost of building the tree is $O(n^{1 + o(1)})$ bit operations.
	
	Now, given a subproduct tree $T$ for $P$, we can compute $D = \{ d \bmod p : p \in P \}$ as 
	follows. We reduce $d$ modulo the root of $T$, and then reduce the resulting value modulo the 
	roots of the left and right subtrees, and so on. The set $D$ is obtained when we reach the 
	leaves. Again, the cost of the reductions at each level of $T$, and hence the total cost of 
	computing $D$, is $O(n^{1 + o(1)})$. From $D$ we get the complete factorization of $d$ at a 
	negligible cost.
\end{proof}
\begin{remark}
	Another way of factoring $d$ in Proposition \ref{prop:factor-d} is to use a quantum factoring 
	algorithm, see \cite{kaye2007introduction}. A nontrivial factor of $d$ can be found in time 
	$O(\log^{2 + o(1)}d)$ bit operations, so a full factorization requires $O(\log^{3 + o(1)}d) 
	= O(n^{1 + o(1)})$ bit operations. 
\end{remark}
\begin{theorem}\label{thm:alg-ddf}
	Let $f \in \F_q[x]$ be a square free polynomial of degree $n$ and let $d$ be the order of the 
	Frobenius of $\F_q[x] / f$. Then the runtime of Algorithm \ref{alg:ddf} is
	\begin{itemize}
		\item $O(n^{1 + o(1)} \log^{2 + o(1)}q)$ bit operations if $d \in O(2^{\log^2n})$, or
		\item $O(n^{4 / 3 + o(1)} \log^{2 + o(1)}q)$ bit operations otherwise.
	\end{itemize}
\end{theorem}
\begin{proof}
	Step \ref{step:sigma-s} of the algorithm takes $O(\CC(m)\log s + \MM(m)\log q) = O(\CC(m)\log m 
	+ \MM(m)\log q)$ operations in $\F_q$ or $O(m^{1 + o(1)}\log^{2 + o(1)}q)$ bit operations. The 
	cost of Step \ref{step:order} is given by Proposition \ref{prop:exact-d}: if $d \in $ 
	$O(2^{\log^2n})$ then it takes $O(m^{1 + o(1)} (\log^2n) \log^{1 + o(1)}q)$ bit operations. 
	Otherwise it takes $O(m^{4 / 3 + o(1)} \log^{2 + o(1)}q)$ bit operations. Step 
	\ref{step:factor-d} can be done in $O(n^{1 + o(1)})$ bit operations according to Proposition 
	\ref{prop:factor-d}. 
	
	Step \ref{step:gcd-first} is done in $O(\CC(m)\log d + \MM(m)\log m)$ operations in $\F_q$ 
	or $O(m^{1 + o(1)}(\log d)\log^{1 + o(1)}q)$ bit operations. For the loop at Step 
	\ref{step:for-ddf}, we first compute the values $\tilde{\sigma}^{d / p_1}, \dots, 
	\tilde{\sigma}^{d / p_\ell} \bmod f$ using Proposition \ref{prop:frobs} and then compute the 
	gcd's. This takes $O(\CC(m)(\log \ell)\log d + \ell\MM(m)\log m)$ operations in $\F_q$. Since 
	the number of prime factors of any integer $t$ is $O(\log t / \log\log t)$ 
	\cite{hardy1979introduction}, we have $\ell \in O(\log d / \log\log d)$ so the cost of this 
	step is $O(m^{1 + o(1)}(\log d)(\log \log d)\log^{1 + o(1)}q)$ bit operations.
	
	If $d \in O(2^{\log^2n})$ then, by the above, for an input polynomial of degree $m$ the 
	algorithm takes 
	\[ U(m) \in O(m^{1 + o(1)}\log^{2 + o(1)}q + m^{1 + o(1)}(\log^2n)(\log \log n)\log^{1 + o(1)}q 
	+ n^{1 + o(1)})
	\]
	bit operations before the recursive calls. Let $d = p_1^{e_1} p_2^{e_2} \cdots p_\ell^{e_\ell}$ 
	be the prime factorization of $d$, and let $T(n)$ be the total cost of the algorithm for an 
	input polynomial $f$ of degree $n$. Then
	\begin{equation}
	\label{equ:cost-ineq}
		T(n) \le \max_{m_1 + \cdots + m_\ell = n} \{ \sum T(m_i) \} + U(n)
	\end{equation}
	where $\{ m_i \}_{1 \le i \le \ell}$ is a partition of irreducible degrees of $f$. Since $U$ is 
	a super-additive function, i.e., $U(m_1 + m_2) \le U(m_1) + U(m_2)$ for all $m_1, m_2$, and by 
	Lemma \ref{lem:depth}, the depth of the recursion is $O(\log n)$, we have
	\[ T(n) \in O(n^{1 + o(1)} \log^{2 + o(1)}q). \]
	If $d \notin O(2^{\log^cn})$ then for an input polynomial of degree $m$ it is always guaranteed 
	that $\log d \in O(m^{1 / 3 + o(1)})$. In this case, the algorithm takes $U(m) = O(m^{4 / 3 + 
	o(1)} \log^{2 + o(1)}q  + n^{1 + o(1)})$ bit operations before the recursive calls. Again since 
	$U$ is a super-additive function and \eqref{equ:cost-ineq} holds, for an input polynomial of 
	degree $n$, the total cost of the algorithm is
	\[ T(n) \in O(n^{4 / 3 + o(1)} \log^{2 + o(1)}q) \]
	bit operations.
\end{proof}

\paragraph{Proof of Theorem \ref{thm:main}.}
For a polynomial $f \in \F_q[x]$ of degree $n$, the squarefree factorization and equal-degree 
factorization stages take $O(n^{1 + o(1)}\log^{2 + o(1)}q)$ bit operations. For the distinct-degree 
factorization, combining Theorems \ref{thm:split-order} and \ref{thm:alg-ddf} gives the desired 
complexity. \hfil \qed

\bibliographystyle{plain}
\bibliography{references}

\end{document}